\documentclass[11pt]{article}
\usepackage{odonnell}

\newcommand{\Bdry}{\Phi}
\newcommand{\Spars}{\mu}
\newcommand{\nullity}{\mathrm{nullity}}
\newcommand{\partsize}{\Delta}

\begin{document}

\title{Markov chain methods for small-set expansion}

\author{Ryan O'Donnell\thanks{Department of Computer Science, Carnegie Mellon University. Supported by NSF grants CCF-0747250 and CCF-0915893, and by a Sloan fellowship.} \and David Witmer\thanks{Department of Computer Science, Carnegie Mellon University.}}

\maketitle

\begin{abstract}
Consider a finite irreducible Markov chain with invariant distribution~$\pi$.  We use the inner product induced by $\pi$ and the associated heat operator to simplify and generalize some results related to graph partitioning and the small-set expansion problem.  For example, Steurer showed a tight connection between the number of small eigenvalues of a graph's Laplacian and the expansion of small sets in that graph. We give a simplified proof which generalizes to the nonregular, directed case. This result implies an approximation algorithm for an ``analytic" version of the Small-Set Expansion Problem, which, in turn, immediately gives an approximation algorithm for Small-Set Expansion.  We also give a simpler proof of a lower bound on the probability that a random walk stays within a set; this result was used in some recent works on finding small sparse cuts.

\ignore{
    Consider an irreducible reversible Markov chain on state space $V$, with $|V| = n$ and invariant distribution~$\pi$.  Let $0 = \lambda_1 \leq \lambda_2 \leq \cdots \lambda_n \leq 2$ be the eigenvalues of its Laplacian operator.  We give a simple spectral condition under which there exists a unit vector $f \in L^2(V,\pi)$ with $\|f\|_1^2 \leq \delta$ and $\la f, Lf \ra \leq \eps$.  (Using a standard Cheeger inequality, this implies the existence of a set $S \subseteq V$ with measure at most $O(\delta)$ and expansion at most $O(\sqrt{\eps})$.) As a consequence we reprove a result of Steurer's stating that if the Laplacian has many small eigenvalues, there exists a small set with small expansion.  We then show that for any $k \in [n]$ and small $\alpha > 0$, there is always a set $S \subseteq V$ with measure at most $O(k^{-1+\alpha})$ and expansion at most $\sqrt{\lambda_k \log_k n} \cdot O(\alpha^{-1/2})$.  We also reprove a lower bound on the probability that a random walk stays within a set using similar spectral methods.
}
\end{abstract}

\section{Overview}

Graph partitioning using spectral methods has recently been the subject of intensive study.  Many results in this area have been proven using discrete-time random walks.  However, these techniques work best when applied to regular graphs with nonnegative eigenvalues.  As a result, it has become standard to move to a lazy version of a graph by adding self-loops, i.e. using $(I + K)/2$ instead of~$K$ as the adjacency matrix.  Much work has also focused on regular graphs only or considered the normalized Laplacian $D^{-1/2}LD^{-1/2}$.

In this work we show that these problems can be avoided using Markov chain techniques, leading to simpler and more general proofs of results related to spectral graph partitioning.  Rather than using discrete-time random walks, we consider continuous-time random walks and the associated heat operator.  ``Smoothing out" the random walk makes the eigenvalues nonnegative, avoiding the need to move to lazy graphs and allowing our techniques to be directly applied to the original instance.  In addition, we use the inner product defined with respect the invariant distribution $\pi$ of the Markov chain representing a random walk on the graph.  We are then able to use our methods directly on nonregular graphs.

We will now give a brief description of some previous results in spectral graph partitioning.  Let $G = (V,E)$ be a graph on $n$ vertices.  Let~$K$ be its (normalized) adjacency matrix, let~$L$ be its (normalized) Laplacian matrix (namely $I - K$), and let $0 = \lambda_1 \leq \lambda_2 \leq \cdots \leq \lambda_n \leq 2$ be the eigenvalues of~$L$.  The conductance $\Bdry[S]$ of a set $S \subseteq V$ is defined to be $\frac{E(S, \overline{S})}{\sum_{v \in S}\deg(v)}$.  The conductance profile of $G$, denoted  $\Phi_G$ , was defined by Lov\'{a}sz and Kannan~\cite{LK99} as
\[
    \Phi_G(r) = \min \{ \Bdry[S] :  S \subseteq V, \mu[S] \leq r\}.
\]

Cheeger's inequality for graphs~\cite{AM85,Alo86,SJ89} states that $V$ can be partitioned into nonempty $S_1$, $S_2$ such that $\Bdry[S_i] \leq O(\sqrt{\lambda_2})$.  Very recently, Louis, Raghavendra, Tetali, and Vempala~\cite{LRTV12} and Lee, Oveis Gharan, and Trevisan~\cite{LOT12}  have given a ``higher order Cheeger inequality'' involving higher eigenvalues.  Specifically, the two results show that for any~$k$, one can partition $V$ into $\Omega(k)$ disjoint nonempty sets $S_i$, each of which has conductance $\Bdry[S_i] \leq O(\sqrt{\lambda_k \log k})$.  Since one of these parts has volume $\mu[S_i] \coloneqq |S_i|/|V| \leq O(1/k)$ we may conclude that
\begin{equation} \label{eqn:LOT}
    \Phi_G(\tfrac{\text{const}}{k}) \leq O(\sqrt{\lambda_k \log k}).
\end{equation}

As noted in these works, for a fixed~$k$ the ``extra factor'' of $\Theta(\sqrt{\log k})$ in~\eqref{eqn:LOT} is necessary; indeed this is true~\cite{LOT12} for all $k \leq \log_2 n$.  However, somewhat intriguingly, the extra factor becomes \emph{unnecessary} once $k$ is as large as $n^{\Omega(1)}$ --- at least, if one is willing to compromise somewhat on the volume parameter.  Specifically, Arora, Barak, and Steurer~\cite{ABS10} showed for regular graphs that
\begin{equation} \label{eqn:ABS}
    \Phi_G(O(k^{-1/100})) \leq O(\sqrt{\lambda_k \log_k n}).
\end{equation}
In his thesis, Steurer~\cite{Ste10a} improved this bound to
\begin{equation}
\Phi_G(k^{-1+1/A}) \leq O(\sqrt{A \lambda_k \log_k n}) \quad \text{for any (sufficiently large) constant $A$.}
\end{equation}
Using Markov chain methods, we give what we feel is a much simpler proof of this result, which also works for the nonregular (and also directed) case.  Our result also implies an approximation algorithm for an ``analytic" version of the Small-Set Expansion problem.  This, in turn, immediately gives an approximation algorithm for Small-Set Expansion by a standard version of Cheeger's Inequality,

In somewhat related recent work, Oveis Gharan and Trevisan~\cite{OT12} proved a weaker version of this bound with $k^{-1/3}$ in place of $k^{-1+1/A}$.  The main point of that work, along with the independent work of Kwok and Lau~\cite{KL12} give a polynomial-time algorithm for the Small-Set Expansion problem in an unweighted (nonregular) graph $G = (V,E)$ with the following guarantee:  if there exists $S \subseteq V$ with $\Spars[S] \leq \delta$ and $\Bdry[S] \leq \eps$, the algorithm finds $T \subseteq V$ with $\Spars[T] \leq O(\delta) \cdot (\delta |E|)^{\alpha}$ and $\Spars[S] \leq O(\sqrt{\eps/\alpha})$ (for any small $\alpha > 0$).  To achieve this, both papers prove a theorem stating that for any $S \subseteq V$ and integer $t > 0$, the probability that a $t$-step random walk starting from a random $x \in S$ stays entirely within $S$ is at least  $\left(1-\frac{\Phi[S]}{2}\right)^t$.  We also give a simpler proof of this result for continuous-time random walks.

\subsection{Our results}

\subsubsection{Bounding the spectral profile}
In this work we provide a different, simple proof of Steurer's improved result using continuous-time random walks instead of lazy discrete-time random walks:
\begin{theorem} \label{thm:our-conductance-profile}
    In any strongly connected graph $G$, $\Phi_G(16 k^{-1+1/A}) \leq 2\sqrt{A} \cdot \sqrt{\lambda_k \log_k n}$ for any real $A \geq 3$.
\end{theorem}
\noindent For example, $\Phi_G(k^{-.999}) \leq O(\sqrt{\lambda_k \log_k n})$ for $k$ sufficiently large.  See Section~\ref{sec:prelims} for the appropriate definitions of $\Phi_G$, $L$, $\lambda_i$, etc.\ in the context of general graphs~$G$.\\

In fact, our result is stronger than this in that we are able to directly bound the \emph{spectral profile} of~$G$.  (The same is true of the result in Arora--Barak--Steurer~\cite{ABS10} and in Steurer's thesis~\cite{Ste10a}.)  Recall that the spectral profile $\Lambda_G$ of $G$, introduced by Goel, Montenegro, and Tetali~\cite{GMT06}, is defined by
\[
    \Lambda_G(r) = \min \left\{ \tfrac{\la f, L f\ra}{\|f\|_2^2} : \text{nonzero } f \co V \to \R^{\geq 0} \text{ with } \pi(\supp(f)) \leq r\right\}.
\]
Goel, Montenegro, and Tetali showed that the ``Cheeger rounding analysis'' yields the following relationship with conductance profile: $\Phi_G(r) \leq \sqrt{2\Lambda_G(r)}$ for all~$r$.\footnote{Actually,~\cite{GMT06} defined $\Lambda_G(r)$ as the minimization of $\tfrac{\la f, L f\ra}{\|f\|_2^2 - \|f\|_1^2}$. But their proof of this relationship still goes through.} As in~\cite{ABS10} we work with a slightly different definition of spectral profile, for technical convenience:
\[
    \Lambda'_G(r) = \min \{ \Bdry[f] : \mu[f] \leq r \}, \qquad \text{where } \quad
    \Bdry[f] = \frac{\la f, L f\ra}{\|f\|_2^2}, \quad \mu[f] = \frac{\|f\|_1^2}{\|f\|_2^2}
\]
are appropriate generalizations of boundary size and volume to functions $f \co V \to \R$.  (These definitions agree with our earlier ones when $f$ is the $0$-$1$ indicator of a set $S \subseteq V$.)  As noted in~\cite[Lemma~A.2]{ABS10} we have $\Lambda_G(4r) \leq 2 \Lambda'_G(r)$ for all~$r$.  (A similar reverse connection also holds.)  Thus:
\begin{theorem} \label{thm:GMT} (Essentially from~\cite{GMT06}.) $\Phi_G(4r) \leq 2\sqrt{\Lambda'_G(r)}$ for all~$r$.
\end{theorem}

We use this connection to obtain Theorem~\ref{thm:our-conductance-profile}; our main theorem is in fact:
\begin{theorem} \label{thm:our-spectral-profile}
    In any strongly connected graph $G$, $\Lambda'_G(4 k^{-1+1/A}) \leq A \cdot \lambda_k \log_k n$ for any real $A \geq 3$.
\end{theorem}
This route to bounding the conductance profile is somewhat in contrast to the works~\cite{LRTV12,LOT12}, both of which combine their spectral analysis and ``rounding algorithm''.\\

Indeed, in this work we consider the ``analytic'' version of the Raghavendra--Steurer~\cite{RS10} Small-Set Expansion problem: given a graph $G = (V,E)$ with the promise that there is a function $f \co V \to \R$ which has $\mu[f] \leq \delta$ and $\Bdry[f] \leq \eps$, find a function $g \co V \to \R$ with $\mu[g] \leq O(\delta)$ and $\Bdry[g]$ as small as possible.  Following~\cite{ABS10}, we provide an eigenspace enumeration lemma which, when combined with Theorem~\ref{thm:our-spectral-profile}, yields the following:
\begin{theorem} \label{thm:our-alg}
    For any $\alpha \leq \frac13$ and $C \geq 1$, there exists an algorithm running in time $\exp(O(n^\alpha) \cdot\tfrac{1}{\delta}\log(C/\delta))$ with the following guarantee:  If there exists $f \co V \to \R$ with $\Spars[f] \leq \delta \leq 1/2$ and $\Bdry[f] \leq \eps \leq 1/4$, the algorithm finds $g \co V \to \R$ with $\Spars[g] \leq \delta \cdot(1+1/C)$ and $\Bdry[g] \leq O(\frac{C^2}{\alpha \delta}) \cdot \eps$.
\end{theorem}

As a byproduct, using Theorem~\ref{thm:GMT} we can immediately deduce the following approximation algorithm for Small-Set Expansion:
\begin{corollary} \label{cor:alg}
    Fix any small constants $\alpha, \delta > 0$.  Then there is an algorithm running in time $\exp(O(n^\alpha))$ with the following guarantee:  If there exists $S \subseteq V$ with $\Spars[S] \leq \delta$ and $\Bdry[S] \leq \eps$, the algorithm finds $T \subseteq V$ with $\Spars[T] \leq 5\delta$ and $\Bdry[T] \leq O(\sqrt{\eps})$.

    More generally, one can obtain $\Bdry[T] \leq O(\eps^{\beta/2})$ in time $\exp(O(n^{\alpha \eps^{1-\beta}}))$ for any $0 < \beta \leq 1$.
\end{corollary}

This result is incomparable with the Arora--Barak--Steurer Small-Set Expansion algorithm: their work had $O(\eps^{\beta/3})$ in place of $O(\eps^{\beta/2})$ and was analyzed only for regular graphs. On the other hand, our Corollary~\ref{cor:alg} holds only for $\delta$ a constant, whereas their algorithm works for $\delta$ as small as~$n^{-\eps^{1-\beta}}$ (which is the more interesting parameter range).

\subsubsection{Continuous-time random walks}
In~\cite{OT12}, Oveis Gharan and Trevisan prove a lower bound on the probability that a random walk stays within a set. (Kwok and Lau~\cite{KL12} prove a similar but somewhat weaker bound.)  Specifically, they show:
\begin{theorem} \label{thm:ot12}
    Let $G = (V,E)$ be an undirected graph with invariant distribution~$\pi$.  Let $\emptyset \neq S \subseteq V$ and let $t > 0$ be an integer.  Choose $\bx \sim \pi$ conditioned on $\bx \in S$, and then perform a $t$-step discrete-time random walk from~$\bx$.  Then the probability that the walk stays entirely within~$S$ is at least $\left(1-\frac{\Phi[S]}{2}\right)^t$.
\end{theorem}

We provide a simple proof of a similar theorem using Markov chain methods.

\begin{theorem}
\label{thm:cont_rand_walk}
In the setting of Theorem~\ref{thm:ot12}, if we instead perform a time-$t$ continuous-time random walk, the probability that the walk stays entirely within~$S$ is at least $\exp(-t\Phi[S])$.
\end{theorem}

\ignore{
Let $G = (V,E)$ be a graph on $n$ vertices which, for this discussion, we assume is undirected and $d$-regular. Let~$K$ be its (normalized) adjacency matrix, let~$L$ be its (normalized) Laplacian matrix (namely $I - K$), and let $0 = \lambda_1 \leq \lambda_2 \leq \cdots \leq \lambda_n \leq 2$ be the eigenvalues of~$L$. It is a simple fact that if $\lambda_2 = 0$ then $G$ is disconnected; i.e., $V$ can be partitioned into two nonempty parts $S_1$, $S_2$, each of which has no edges on its boundary.  Somewhat less simple is \emph{Cheeger's inequality} for graphs~\cite{AM85,Alo86,SJ89}, which gives a robust version of this fact: if $\lambda_2$ is small then $V$ can be partitioned into nonempty $S_1$, $S_2$, each of which has only a small fraction of its edges on its boundary.  More precisely, each $S_i$ has conductance $\Bdry[S_i] \coloneqq \frac{E(S_i, \overline{S_i})}{d |S_i|} \leq O(\sqrt{\lambda_2})$.

Regarding higher eigenvalues, another simple fact is that $\lambda_k = 0$ if and only if $G$ has at least~$k$ connected components.  (In other words, $\nullity(L)$ is the number of connected components of~$G$.)  It is natural to ask if there is an associated ``higher order Cheeger inequality''.  Positive results in this direction were recently obtained by Louis, Raghavendra, Tetali, and Vempala~\cite{LRTV11,LRTV12} and by Lee, Oveis Gharan, and Trevisan~\cite{LOT12}.  Specifically, the latter two results show that for any~$k$, one can partition $V$ into $\Omega(k)$ disjoint nonempty sets $S_i$, each of which has conductance $\Bdry[S_i] \leq O(\sqrt{\lambda_k \log k})$.  Since one of these parts has volume $\mu[S_i] \coloneqq |S_i|/|V| \leq O(1/k)$ we may conclude that
\begin{equation} \label{eqn:LOT}
    \Phi_G(\tfrac{\text{const}}{k}) \leq O(\sqrt{\lambda_k \log k}),
\end{equation}
where $\Phi_G$ is the \emph{conductance profile} of $G$, defined by Lov\'{a}sz and Kannan~\cite{LK99} as
\[
    \Phi_G(r) = \min \{ \Bdry[S] :  S \subseteq V, \mu[S] \leq r\}.
\]
(In fact,~\cite{LOT12} shows that ``const'' may be arbitrarily close to~$1$.)  We remark that the results of~\cite{LOT12,LRTV12} are shown in the more general context of arbitrary edge-weighted undirected graphs.

As noted in these works, for a fixed~$k$ the ``extra factor'' of $\Theta(\sqrt{\log k})$ in~\eqref{eqn:LOT} is necessary; indeed this is true~\cite{LOT12} for all $k \leq \log_2 n$.  However, somewhat intriguingly, the extra factor becomes \emph{unnecessary} once $k$ is as large as $n^{\Omega(1)}$ --- at least, if one is willing to compromise somewhat on the volume parameter.  Specifically, Arora, Barak, and Steurer~\cite{ABS10} showed (for regular graphs) that
\begin{equation} \label{eqn:ABS}
    \Phi_G(O(k^{-1/100})) \leq O(\sqrt{\lambda_k \log_k n}).
\end{equation}
This was the key technical tool needed for their subexponential-time algorithms for Small-Set Expansion and Unique Games.  Regarding the parameter $k^{-1/100}$, Arora, Barak, and Steurer wrote that they ``do not know if the constant~$100$ [can] be replaced with $1 + o(1)$ (though such a strong bound, if true, will require a different proof)''.

In his thesis, Steurer~\cite{Ste10a} essentially resolves this question by showing that
\begin{equation}
\Phi_G(k^{-1+1/A}) \leq O(\sqrt{A \lambda_k \log_k n}).
\end{equation}

\subsection{Our results}
In this work we provide a different, simple proof of Steurer's improved result using continuous-time random walks instead of lazy discrete-time random walks:
\begin{theorem} \label{thm:our-conductance-profile}
    In any strongly connected graph $G$, $\Phi_G(16 k^{-1+1/A}) \leq 2\sqrt{A} \cdot \sqrt{\lambda_k \log_k n}$ for any real $A \geq 3$.
\end{theorem}
\noindent For example, $\Phi_G(k^{-.999}) \leq O(\sqrt{\lambda_k \log_k n})$ for $k$ sufficiently large.  See Section~\ref{sec:prelims} for the appropriate definitions of $\Phi_G$, $L$, $\lambda_i$, etc.\ in the context of general graphs~$G$.\\

In fact, our result is stronger than this in that we are able to directly bound the \emph{spectral profile} of~$G$.  (The same is true of the~\cite{ABS10} and Steurer's CITE results.)  Recall that the spectral profile $\Lambda_G$ of $G$, introduced by Goel, Montenegro, and Tetali~\cite{GMT06}, is defined by
\[
    \Lambda_G(r) = \min \left\{ \tfrac{\la f, L f\ra}{\|f\|_2^2} : \text{nonzero } f \co V \to \R^{\geq 0} \text{ with } \pi(\supp(f)) \leq r\right\}.
\]
Goel, Montenegro, and Tetali showed that the ``Cheeger rounding analysis'' yields the following relationship with conductance profile: $\Phi_G(r) \leq \sqrt{2\Lambda_G(r)}$ for all~$r$.\footnote{Actually,~\cite{GMT06} defined $\Lambda_G(r)$ as the minimization of $\tfrac{\la f, L f\ra}{\|f\|_2^2 - \|f\|_1^2}$. But their proof of this relationship still goes through.} As in~\cite{ABS10} we work with a slightly different definition of spectral profile, for technical convenience:
\[
    \Lambda'_G(r) = \min \{ \Bdry[f] : \mu[f] \leq r \}, \qquad \text{where } \quad
    \Bdry[f] = \frac{\la f, L f\ra}{\|f\|_2^2}, \quad \mu[f] = \frac{\|f\|_1^2}{\|f\|_2^2}
\]
are appropriate generalizations of boundary size and volume to functions $f \co V \to \R$.  (These definitions agree with our earlier ones when $f$ is the $0$-$1$ indicator of a set $S \subseteq V$.)  As noted in~\cite[Lemma~A.2]{ABS10} we have $\Lambda_G(4r) \leq 2 \Lambda'_G(r)$ for all~$r$.  (A similar reverse connection also holds.)  Thus:
\begin{theorem} \label{thm:GMT} (Essentially from~\cite{GMT06}.) $\Phi_G(4r) \leq 2\sqrt{\Lambda'_G(r)}$ for all~$r$.
\end{theorem}

We use this connection to obtain Theorem~\ref{thm:our-conductance-profile}; our main theorem is in fact:
\begin{theorem} \label{thm:our-spectral-profile}
    In any strongly connected graph $G$, $\Lambda'_G(4 k^{-1+1/A}) \leq A \cdot \lambda_k \log_k n$ for any real $A \geq 3$.
\end{theorem}
This route to bounding the conductance profile is somewhat in contrast to the works~\cite{LRTV12,LOT12}, both of which combine their spectral analysis and ``rounding algorithm''.\\

Indeed, in this work we consider the ``analytic'' version of the Raghavendra--Steurer~\cite{RS10} Small-Set Expansion problem: given a graph $G = (V,E)$ with the promise that there is a function $f \co V \to \R$ which has $\mu[f] \leq \delta$ and $\Bdry[f] \leq \eps$, find a function $g \co V \to \R$ with $\mu[g] \leq O(\delta)$ and $\Bdry[g]$ as small as possible.  Following~\cite{ABS10}, we provide an eigenspace enumeration lemma which, when combined with Theorem~\ref{thm:our-spectral-profile}, yields the following:
\begin{theorem} \label{thm:our-alg}
    For any $\alpha \leq \frac13$ and $C \geq 1$, there exists an algorithm running in time $\exp(O(n^\alpha) \cdot\tfrac{1}{\delta}\log(C/\delta))$ with the following guarantee:  If there exists $f \co V \to \R$ with $\Spars[f] \leq \delta \leq 1/2$ and $\Bdry[f] \leq \eps \leq 1/4$, the algorithm finds $g \co V \to \R$ with $\Spars[g] \leq \delta \cdot(1+1/C)$ and $\Bdry[g] \leq O(\frac{C^2}{\alpha \delta}) \cdot \eps$.
\end{theorem}

As a byproduct, using Theorem~\ref{thm:GMT} we can immediately deduce the following approximation algorithm for Small-Set Expansion:
\begin{corollary} \label{cor:alg}
    Fix any small constants $\alpha, \delta > 0$.  Then there is an algorithm running in time $\exp(O(n^\alpha))$ with the following guarantee:  If there exists $S \subseteq V$ with $\Spars[S] \leq \delta$ and $\Bdry[S] \leq \eps$, the algorithm finds $T \subseteq V$ with $\Spars[T] \leq 5\delta$ and $\Bdry[T] \leq O(\sqrt{\eps})$.

    More generally, one can obtain $\Bdry[T] \leq O(\eps^{\beta/2})$ in time $\exp(O(n^{\alpha \eps^{1-\beta}}))$ for any $0 < \beta \leq 1$.
\end{corollary}

This result is incomparable with the Arora--Barak--Steurer Small-Set Expansion algorithm: their work had $O(\eps^{\beta/3})$ in place of $O(\eps^{\beta/2})$ and was analyzed only for regular graphs. On the other hand, our Corollary~\ref{cor:alg} holds only for $\delta$ a constant, whereas their algorithm works for $\delta$ as small as~$n^{-\eps^{1-\beta}}$ (which is the more interesting parameter range).

\subsection{Simultaneous work}
Independently of our work, Oveis Gharan and Trevisan~\cite{OT12} have obtained some results related to the ones in this paper.  They have proven a weaker version of our Theorem~\ref{thm:our-conductance-profile}, with $k^{-1/3}$ in place of our~$k^{-1+1/A}$.  However it seems quite plausible their proof technique would achieve the same result as Theorem~\ref{thm:our-conductance-profile} without much additional work.  They also give a polynomial-time algorithm for the Small-Set Expansion problem in an unweighted (nonregular) graph $G = (V,E)$ with the following guarantee:  if there exists $S \subseteq V$ with $\Spars[S] \leq \delta$ and $\Bdry[S] \leq \eps$, the algorithm finds $T \subseteq V$ with $\Spars[T] \leq O(\delta) \cdot (\delta |E|)^{\alpha}$ and $\Spars[S] \leq O(\sqrt{\eps/\alpha})$ (for any small $\alpha > 0$).
}

\section{Preliminaries}                 \label{sec:prelims}
Instead of directed graphs, we will use the language of Markov chains; for background, see e.g.~\cite{DS96,MT06}.

Throughout this work, $G$ will denote an irreducible Markov chain on state space~$V$ of cardinality~$n$, with no isolated states.  We will be considering elements $f$ in the vector space of functions $V \to \R$. We  write $K$ for the adjacency matrix operator: $Kf(x) = \E_{\by \sim x}[f(\by)]$, where $\by \sim x$ denotes that $\by$ is obtained by taking one step from~$x$ in the chain.  $K$ has a unique invariant probability distribution $\pi$ on $V$ which is nowhere~$0$. It gives rise to an inner product on functions, $\la f, g \ra = \Ex_{\bx \sim \pi}[f(\bx)g(\bx)]$.  We write $L = id - K$  for the Laplacian operator and $H_t = \exp(-tL)$ for the heat kernel (continuous-time transition) operator.

\begin{definition}
    Given nonzero $f \co V \to \R$ we define its \emph{analytic boundary size/conductance} to be
    \[
        \Bdry[f] = \frac{\la f, L f \ra}{\la f, f \ra} = 1 - \frac{\la f, K f \ra}{\la f, f \ra}.
    \]
    Note that if $f$ is the $0$-$1$ indicator of a set $S \subseteq V$ then $\Phi[f] = \Pr_{\bx \sim \pi, \by \sim \bx}[\by \not \in S \mid \bx \in S]$.  We will also write $\Phi[S]$ in this case.
\end{definition}

\begin{definition}
    Given a nonzero $f \co V \to \R$ we define its \emph{analytic sparsity} to be
    \[
        \Spars[f] = \frac{\|f\|_1^2}{\|f\|_2^2}.
    \]
    Note that if $f$ is the $0$-$1$ indicator of a set $S \subseteq V$ then $\Spars[f] = \pi(S)$.
\end{definition}

These definitions motivate consideration of an ``analytic'' version of the Small-Set Expansion Problem: Assuming there is an analytically sparse~$f$ with small analytic boundary, find such an~$f$.  More precisely:

\paragraph{Analytic Small-Set Expansion Problem:} Given as input $G$ with the promise that there exists $f \co V \to \R$ with $\Spars[f] \leq \delta \leq 1/2$ and $\Bdry[f] \leq \eps$, find $f' \co V \to \R$ with $\Spars[f'] \leq \delta'$ and $\Bdry[f'] \leq \eps'$.  In this bicriteria problem, we typically insist that $\delta' = O(\delta)$ and then try to minimize $\eps'$. \\

\noindent Note that the standard Small-Set Expansion problem is the above problem with the additional restriction that~$f$ and~$f'$ should be $0$-$1$-valued functions.  \\

For the remainder of this work we will assume that~$G$ is reversible.  However, this is without loss of generality since, given a non-reversible Markov chain~$G'$ with adjacency matrix operator~$K'$, we can replace it with the reversible Markov chain~$G$ having adjacency matrix operator $K = \frac{K' + {K'}^*}{2}$.  The chain $G$ has the same invariant distribution~$\pi$ as~$G'$ which means that the notion of analytic sparsity is unchanged.  Further, if $L$ and $L'$ are the Laplacians of $G$ and $G'$, respectively, then $\la f, L f \ra = \la f, L' f \ra$ for any $f \co V \to \R$; hence the notion of analytic boundary is also unchanged.

Given a reversible chain~$G$, the operators $K$, $L$, and $H_t$ have a common orthogonal basis of eigenfunctions.  We will write $0 = \lambda_1 \leq \lambda_2 \leq \cdots \leq \lambda_{n}$ for the eigenvalues of~$L$; note that the $i$th eigenvalue of $K$ is $1 - \lambda_i$ and the $i$th eigenvalue of $H_t$ is $\exp(-t\lambda_i)$.  All of our theorems which mention the eigenvalues~$\lambda_i$ hold also for non-reversible chains~$G'$, with the $\lambda_i$'s being those for the associated reversible chain~$G$.

Following~\cite{ABS10}, our algorithm for the Analytic Small-Set Expansion problem (Theorem~\ref{thm:our-alg}) breaks into two cases, depending on the ``analytic nullity'' of~$L$ (called ``threshold rank'' in~\cite{ABS10}):
\begin{definition}
    We define $\nullity_\eta(L) = \#\{i : \lambda_i \leq \eta\}$.  Note that $\nullity_0(L)$ is the usual nullity.
\end{definition}

\begin{remark}
Throughout we will present algorithms in the model of exact arithmetic. E.g., we will assume that given~$G$, the eigenvalues and eigenfunctions of $L$ can be computed exactly.  We believe (but have not verified) that our results can be extended to standard computational models (e.g., Turing machines).
\end{remark}

\section{A new bound on the spectral profile}

Here we give our new spectral criterion, based on the trace of the heat kernel, which ensures the existence of an analytically sparse function with small analytic boundary.
\begin{theorem}                                     \label{thm:heat-algorithm}
    Fix  $0 < \gamma \leq 1 \leq \partsize$ and suppose there exists $t > 0$ such that
    \begin{equation} \label{eqn:key}
        \tr(H_t) - \tfrac{1}{\gamma} \tr(LH_t) \geq \partsize.
    \end{equation}
    Then in $\poly(n)$ time one can find $g \co V \to \R^{\geq 0}$ satisfying $\Spars[g] \leq 1/\partsize$ and $\Bdry[g] \leq \gamma$.
\end{theorem}
\begin{proof}
    Let $\phi_x = \tfrac{1}{\pi(x)} \cdot 1_x$ for $x \in V$, so $\E[\phi_x]= 1$.  Write $\phi'_x = \sqrt{\pi(x)} \cdot \phi_x$, so the collection $(\phi'_x)_{x \in V}$ forms an orthonormal basis.  Since trace is ``the sum of the diagonal entries'', we have
    \[
        \tr(H_t) = \sum_{x \in V} \la \phi'_x, H_t \phi'_x \ra = \sum_{x \in V} \pi(x) \la \phi_x, H_t \phi_x \ra = \E_{\bx \sim \pi} \la H_{t/2} \phi_x, H_{t/2} \phi_x \ra.
    \]
    Similarly, $\tr(L H_t) = \E_{\bx \sim \pi}[ \la H_{t/2} \phi_x, L H_{t/2} \phi_x \ra]$.
    Thus the assumption~\eqref{eqn:key} implies
    \[
        \E_{\bx \sim \pi}[\la H_{t/2} \phi_x, H_{t/2} \phi_x \ra - \tfrac{1}{\gamma} \la H_{t/2} \phi_x, L H_{t/2} \phi_x \ra] \geq \partsize.
    \]
    Select (in $\poly(n)$ time) a particular $x_0 \in V$ achieving at least~$\partsize$ in this expectation.  We define $g = H_{t/2} \phi_{x_0}$ and therefore we have
    \begin{equation} \label{eqn:heat-conc}
        \la g, g \ra - \tfrac{1}{\gamma} \la g, L g \ra \geq \partsize.
    \end{equation}
    Note that $g \geq 0$ since $\phi_{x_0} \geq 0$ and $H_{t/2}$ is positivity-preserving.  Thus $\|g\|_1 = \E[g] = \E[\phi_{x_0}] = 1$.  Further, from~\eqref{eqn:heat-conc} we deduce $\la g, g \ra \geq \partsize$; thus $\Spars[g] \leq 1/\partsize$ as desired.  Finally,~\eqref{eqn:heat-conc} certainly implies $\la g, g \ra - \tfrac{1}{\gamma} \la g, L g \ra \geq 0$, which is equivalent to $\Bdry[g] \leq \gamma$.
\end{proof}

A straightforward calculation now shows that if $L$ has large analytic nullity then we can get good bounds from Theorem~\ref{thm:heat-algorithm}:
\begin{corollary}                                       \label{cor:nullity-to-trace}
    Fix $0 < \gamma \leq 1$.  Let $0 < \alpha \leq \frac13$ and let $k = \nullity_{\alpha \gamma}(L)$.  Assume $k \geq \frac{n^\alpha}{\ln n}$.  Then in $\poly(n)$ time one can find $g \co V \to \R^{\geq 0}$ satisfying $\Bdry[g] \leq \gamma$ and $\Spars[g] \leq 1/\partsize$, where $\partsize = \frac{k}{4n^{\alpha}}$.
\end{corollary}
\begin{proof}
    We show that~\eqref{eqn:key} from Theorem~\ref{thm:heat-algorithm} holds with $\gamma$, $\partsize$, and $t = \frac{1}{\gamma} \ln n$.  We have
    \begin{equation} \label{eqn:calc1}
        \tr(H_t) - \tfrac{1}{\gamma} \tr(LH_t) = \littlesum_{i=1}^{n} (1 - \tfrac{\lambda_i}{\gamma}) \exp(-t\lambda_i)  = \littlesum_{i=1}^{n} (1-\tfrac{\lambda_i}{\gamma}) n^{-\lambda_i/\gamma}.
    \end{equation}
    The expression $(1-r)n^{-r}$ is decreasing for $r \in [0,1]$; for larger~$r$, it attains its minimum at $r = 1+\frac{1}{\ln n}$, where it has value $-\frac{1}{en \ln n}$.  Thus by distinguishing $r =\frac{\lambda_i}{\gamma} \gtrless \alpha$ in~\eqref{eqn:calc1} we may obtain
    \[
        \eqref{eqn:calc1} \geq \#\{i : \lambda_i \leq \alpha \gamma\} \cdot (1-\alpha)n^{-\alpha} - \#\{i : \lambda_i > \alpha \gamma\} \cdot \tfrac{1}{en\ln n} \geq \frac{k}{n^\alpha}(1-\alpha) - \tfrac{1}{e\ln n}.
    \]
    Using $\alpha \leq \frac13$ and $k \geq \frac{n^\alpha}{\ln n}$, the above is indeed at least $\partsize = \frac{k}{4n^{\alpha}}$.
\end{proof}

Restating the parameters yields:
\begin{corollary} \label{cor:alt-nullity-to-trace}
    Let $0 < \delta \leq 1$. If there exists $\alpha \leq \frac13$ such that $\nullity_{\alpha \gamma}(L) \geq \frac{4}{\delta} n^\alpha$, then in  $\poly(n)$ time one can find $g \co V \to \R^{\geq 0}$ satisfying $\Spars[g] \leq \delta$ and $\Bdry[g] \leq \gamma$.
\end{corollary}

An alternative restatement of the parameters yields our main Theorem~\ref{thm:our-spectral-profile}:  simply take $\alpha = \tfrac{1}{A \log_k n}$ and $\gamma = A \lambda_{k} \log_k n$ in Corollary~\ref{cor:nullity-to-trace}.

\section{An algorithm for Analytic Small-Set Expansion}

In~\cite{ABS10} it is shown that when $L$ has small analytic nullity, one can find sparse sets by brute-force search through low-eigenvalue eigenspace.
We present a very similar algorithm for finding analytically sparse sets.
\begin{lemma}                                     \label{lem:low-nullity}
    Suppose there exists $f \co V \to \R$ with
    \[
        \Spars[f] \leq \delta \leq 1/2, \qquad \Bdry[f] \leq \eps \leq 1/4.
    \]
    Let $2\eps \leq \eta \leq 1$.  Then in time $\exp(O(\nullity_\eta(L) \log(\eta/\eps)))\cdot \poly(n)$ one can find $g \co V \to \R$ satisfying
    \[
        \Spars[g] \leq \delta + O(\eps/\eta + \sqrt{\delta \eps/\eta}) \leq O(\delta + \eps/\eta), \qquad \Bdry[g] \leq \eta.
    \]
\end{lemma}
\begin{remark}  It is also quite easy to show $g$ will satisfy $\Bdry[g] \leq O(\sqrt{\eps/\eta})$, which is useful if $\eta \gg \eps^{1/3}$.  We will not need this parameter setting, so we omit the proof. \noteryan{it's in the tex, commented out} \ignore{It remains to show that also $\Bdry[g] = \la g, L g\ra \leq \eps + O(\sqrt{\eps/\eta})$. For this,
    \[
        \la g,Lg \ra = \la f,Lf \ra + \la f,L(g-f) \ra + \la g-f,Lg \ra \leq \eps + \|L(g-f)\|_2 + \|g-f\|_2\|Lg\|_2,
    \]
    where we used $\Bdry[g] \leq \eps$, Cauchy--Schwarz twice, and $\|f\|_2 = 1$.  We also have $\|L g\|_2 \leq \eta$ (since $g \in U$ is unit) and $\|L(g-f)\|_2 \leq 2\|g-f\|_2$ (since $L$'s eigenvalues are at most~$2$).  Thus
    \[
        \la g, Lg \ra \leq \eps + (2+\eta)\|f-g\|_2 \leq \eps + 6\sqrt{\eps/\eta},
    \]
    as needed.}
\end{remark}
\begin{proof}
    Let $\psi_1, \dots, \psi_m$ be an orthonormal basis of eigenfunctions for $L$, corresponding to eigenvalues $\lambda_1, \dots, \lambda_n$.  Without loss of generality, assume $\|f\|_2 = 1$.  Write $m = \nullity_\eta(L)$ and write $U$ for the dimension-$m$ subspace spanned by $\psi_1, \dots, \psi_m$.  Express $f = \sum_{i=1}^{n} c_i \psi_i$, so $\sum c_i^2 =1$ by the orthonormality of the $\psi_i$'s.  We have
    \[
        \eps \geq \Bdry[f] = \la f, L f \ra = \sum_{i=1}^{n} \lambda_i c_i^2 \geq \sum_{i> m} \lambda_i c_i^2 \geq \eta \sum_{i > m} c_i^2. 
    \]
    In other words, if $f_U$ denotes $\sum_{i \leq m} c_i \psi_i $ then $\|f - f_U\|_2^2 \leq \eps/\eta$ (which is at most $1/2$ by the assumption on $\eta$).  If we define $u \in U$ to be the unit vector $f_U/\|f_U\|_2$, it follows that
    \[
        \|f - u\|_2 \leq \sqrt{2\eps/\eta}.
    \]
    As in~\cite{ABS10} we can now consider all $g$ in a $.5\sqrt{\eps/\eta}$-net for the unit sphere of~$U$.\noteryan{Would it be simpler to do a net on the ball?  And not normalize $f_U$?} The cardinality of this net is $\exp(O(m \log(\eta/\eps)))$.  One such $g$ will satisfy
    \[
        \|u - g\|_2 \leq .5\sqrt{\eps/\eta} \quad \text{and hence} \quad \|f - g\|_2 \leq 2\sqrt{\eps/\eta}.
    \]
    For this $g$ we have
    \[
        \|g\|_1 \leq \|f\|_1 + \|f - g\|_1 \leq \sqrt{\Spars[f]} + \|f-g\|_2 \leq \sqrt{\delta} + 2\sqrt{\eps/\eta}
    \]
    and hence $\Spars[g] \leq \delta + O(\eps/\eta + \sqrt{\delta \eps/\eta})$, as desired.  Since $g$ is a unit vector in $U$ we may also immediately conclude $\Bdry[g] \leq \eta$.
\end{proof}

From Corollary~\ref{cor:alt-nullity-to-trace} we know that if $L$ has large analytic nullity then there is automatically an (easily findable) $f \co V \to \R$ which is analytically sparse and has small analytic boundary.  On the other hand, if $L$ has small analytic nullity, the above lemma can solve the Analytic Small-Set Expansion problem in not too much time.  Combining these facts lets us prove our Theorem~\ref{thm:our-alg}, restated here for convenience:
\paragraph{Theorem~\ref{thm:our-alg}.} \emph{For any $\alpha \leq \frac13$ and $C \geq 1$, there exists an algorithm running in time $\exp(O(n^\alpha) \cdot\tfrac{1}{\delta}\log(C/\delta))$ with the following guarantee:  If there exists $f \co V \to \R$ with $\Spars[f] \leq \delta \leq 1/2$ and $\Bdry[f] \leq \eps \leq 1/4$, the algorithm finds $g \co V \to \R$ with $\Spars[g] \leq \delta \cdot(1+1/C)$ and $\Bdry[g] \leq O(\frac{C^2}{\alpha \delta}) \cdot \eps$.}

\begin{proof}
    Set $\gamma = \frac{B}{\alpha \delta} \cdot \eps$; we will eventually take $B = O(C^2)$. If $\nullity_{\alpha \gamma}(L) \geq \frac{4}{\delta} n^{\alpha}$ then from Corollary~\ref{cor:alt-nullity-to-trace} we can find $g$ with $\Spars[g] \leq \delta$, $\Bdry[g] \leq \gamma$ in $\poly(n)$ time; in fact, here we don't even need to assume the existence of~$f$. Otherwise, Lemma~\ref{lem:low-nullity} tells us that in time $\exp(O(n^\alpha) \cdot\tfrac{1}{\delta}\log(B/\delta))$ we can find a $g$ satisfying
    \[
        \Spars[g] \leq \delta + O(\tfrac{\delta}{B} + \sqrt{\tfrac{\delta^2}{B}}) = \delta \cdot (1+O(1/\sqrt{B})), \qquad \Bdry[g] \leq \alpha \gamma \leq \gamma.
    \]
    Thus the result follows by taking $B = O(C^2)$.
\end{proof}

\ignore{
\section{Algorithm for small-set expansion}
Let us now combine these results to give an algorithm for the Analytic Expansion Problem.

We can combine these results to solve the small set expansion problem.  Assume there exists a set $S$ with $\pi(S) \leq \delta$ and $\Phi(S) \leq \varepsilon$.

If $\mathrm{rank}_{ \eta}(L) \geq n^{3\eta/\gamma}$, we are in the high-rank case and can find a set $S'$ such that $\pi(S') =  O(\delta)$ and $\Phi(S') = O(\sqrt{\gamma})$ in polynomial time.  If  $\mathrm{rank}_{ \eta}(L) < n^{3\eta/\gamma}$, we can find a set $S'$ such that $\pi(S') = \delta + O(\varepsilon/\eta)$ and $\Phi(S') = \min(\sqrt{\eta}, O((\varepsilon/\eta)^{1/4})$ in time $\exp(n^{3\eta/\gamma}\log 1/\varepsilon)$.

We can recover the result of [ABS] Theorem 2.1.  For $\alpha \in (0,1)$, we can set $\eta = \varepsilon^{1-\alpha/3}$ and $\gamma = \varepsilon^{2\alpha/3}$ to get a set $S'$ such that $\pi(S') = O(\delta)$ and $\Phi(S') = O(\varepsilon^{\alpha/3})$ in time $\exp(n^{O(1-\alpha)}\log 1/\varepsilon)$.  We actually do slightly better in the low-rank case, getting $S'$ such that $\Phi(S') = O(\varepsilon^{1/2-\alpha/6}) \leq O(\varepsilon^{\alpha/3})$.

Alternatively, we can set $\eta = \varepsilon/\delta$ and $\gamma = \varepsilon/\delta^2$ to get $S'$ such that $\pi(S') = O(\delta)$ and $\Phi(S') = O(\sqrt{\varepsilon}/\delta)$ in time $\exp(n^{O(\delta)}\log 1/\varepsilon)$.
}

\section{The probability a random walk stays entirely within a set}

In \cite{OT12} the authors show that a $t$-step discrete time random walk starting from a random vertex in $S \subseteq V$ stays entirely within $S$ with probability at least $\left(1-\frac{\Phi[S]}{2}\right)^t$.  We give a proof of a similar result for continuous-time random walks using Markov chain methods.

\paragraph{Theorem~\ref{thm:cont_rand_walk} restated.} \emph{For any $\emptyset \neq S \subseteq V$ and real $t > 0$, let $C(t,S)$ denote the probability that a continuous-time-$t$ random walk, started from a random $\bx \sim S$, stays entirely within~$S$.  Then $C(t,S) \geq \exp(-t\Phi[S])$.}

\begin{proof}
Let us define an operator $K_S$ on functions $f \colon V \to \R$ as follows:  $K_Sf(x) = \E_{\by \sim x}[1_S(x)1_S(\by)f(\by)]$, where $1_S$ is the indicator function for $S$.  It is easy to see that~$K_S$ is self-adjoint; thus it has~$n$ real eigenvalues and~$n$ linearly independent eigenvectors.  We also define $L_S = id - K_S$ and $H_{t,S} = \exp(-tL_S)$.  Let $v_1, \dots, v_n$ be an orthonormal basis of eigenvectors of~$L_S$ (which are also eigenvectors of $K_S$ and $H_{t,S}$) and let $\lambda_1, \dots, \lambda_n$ be the corresponding eigenvalues of $L_S$. Finally, we define $\phi'_S = \tfrac{1}{\sqrt{\pi(S)}} \cdot 1_S$ and write $\phi'_S = \sum_i c_i v_i$ for some constants~$c_i$.  Since $\|\phi'_S\|_2 = 1$ it follows that $\sum_i c_i^2 = 1$:

First, we will show that $\Phi[S] =  \sum_i c_i^2 \lambda_i$.
\begin{align*}
\Phi[S] &= \Pr_{\substack{\bx \sim \pi \\ \by \sim \bx}}[\by \notin S \mid \bx \in S] \\
             &= \Pr_{\substack{\bx \sim \pi \\ \by \sim \bx}}[\by \notin S \wedge \bx \in S] / \Pr_{\bx \sim \pi}[\bx \in S] \\
             &= \tfrac{1}{\pi(S)}\E_{\substack{\bx \sim \pi \\ \by \sim \bx}}[1_S(\bx)(1_S(\bx)-1_S(\by))] \\
             &= \tfrac{1}{\pi(S)}\E_{\substack{\bx \sim \pi \\ \by \sim \bx}}[1_S(\bx)(1_S(\bx)-1_S(\bx) 1_S(\by))] \\
             &= \tfrac{1}{\pi(S)}\E_{\bx \sim \pi}[1_S(\bx)(1_S(\bx)-1_S(\bx)\E_{\by \sim \bx}[1_S(\by)])] \\
             &= \tfrac{1}{\pi(S)}\E_{\bx \sim \pi}[1_S(\bx)(id1_S(\bx)-K_S1_S(\bx))] \\
             &= \la \phi'_S, L_S \phi'_S \ra \\
             &= \sum_i c_i^2 \lambda_i.
\end{align*}

Now we show that $C(t,S) = \sum_i c_i^2 \exp(-t \lambda_i)$.  Let $\bw_0, \dots, \bw_{\btau}$ be the states of a time-$t$ continuous-time random walk in~$G$; note that this is the same as a $\btau$-step discrete-time random walk, where $\btau \sim \mathrm{Poisson}(t)$.  Let $\bW$ denote the set of all states visited.  Then:
\begin{align*}
C(t,S) &= \Pr[\bW \subseteq S \mid \bw_0 \in S] \\
           &= \tfrac{1}{\pi(S)}\E[1_S(\bw_0) 1_S(\bw_1) \ldots 1_S(\bw_{\btau})] \\
           &= \tfrac{1}{\pi(S)} \E_{\bx \sim \pi}\E_{\btau \sim \mathrm{Poisson}(t)}[1_S(\bx)K^{\btau}_S 1_S(\bx)] \\
           &= \tfrac{1}{\pi(S)}\E_{\bx \sim \pi}[1_S(\bx)H_{t,S}1_S(\bx)] \\
           &= \la \phi'_S, H_{t,S} \phi'_S \ra \\
           &= \sum_i c_i^2 \exp(-t \lambda_i).
\end{align*}

To complete the proof, we need to show that $\sum_i c_i^2 \exp(-t \lambda_i) \geq \exp(-t \sum_i c_i^2 \lambda_i)$.  This follows immediately by the convexity of the exponential function and Jensen's inequality.
\ignore{
\begin{align*}
C(t,S) &= \sum_i c_i^2 \exp(-t \lambda_i) \\
           &\geq \exp(-t \sum_i c_i^2 \lambda_i) \\
           &= \exp(-t \Phi[S]).
\end{align*}
\qedhere
}
\end{proof}

\section*{Acknowledgments}
We thank James Lee, David Steurer, Yu Wu, and Yuan Zhou for helpful discussions.

\bibliographystyle{alpha}
\bibliography{/dropbox/bib/odonnell-bib}

\end{document}